\documentclass[pdflatex,sn-basic]{sn-jnl}% Basic Springer Nature Reference Style/Chemistry Reference Style
\usepackage{graphicx}%
\usepackage{multirow}%
\usepackage{amssymb}
\usepackage{amsmath,amsthm,amscd,amssymb,amsfonts}
\usepackage[title]{appendix}%
\usepackage{xcolor}%
\usepackage{textcomp}%
\usepackage{manyfoot}%
\usepackage{booktabs}%
\usepackage{algorithm}%
\usepackage{algorithmicx}%
\usepackage{algpseudocode}%
\usepackage{listings}%
\usepackage{graphics}%
\theoremstyle{thmstyletwo}%

\theoremstyle{thmstylethree}%

\numberwithin{equation}{section}
\newtheorem{thm}{Theorem}

\newtheorem{cor}{Corollary}

\begin{document}

\title[Extended astrophysical thermonuclear functions]{Closed-form representations of extended astrophysical thermonuclear functions}

%%=============================================================%%
%% GivenName	-> \fnm{Joergen W.}
%% Particle	-> \spfx{van der} -> surname prefix
%% FamilyName	-> \sur{Ploeg}
%% Suffix	-> \sfx{IV}
%% \author*[1,2]{\fnm{Joergen W.} \spfx{van der} \sur{Ploeg}
%%  \sfx{IV}}\email{iauthor@gmail.com}
%%=============================================================%%

\author[1]{\fnm{Subrat} \sur{Parida}}\email{paridasubrat553@gmail.com}\equalcont{Currently Looking for Doctoral Research.}

\affil[1]{\orgdiv{Department of Mathematics}, \orgname{Ramanujan School of Mathematical Sciences, Pondicherry University}, \orgaddress{\city{Pondicherry}, \postcode{605014}, \state{Puducherry}, \country{India}}}

%%==================================%%
%% Sample for unstructured abstract %%
%%==================================%%

\abstract{Reaction rate determination is an essential problem in the theories of cosmological and stellar nucleosynthesis. I investigate the closed-form analytic evaluation of thermonuclear reaction rates, which are necessary for both cosmic and celestial nucleosynthesis. An alternate velocity distribution to the Maxwell-Boltzmann distribution can be considered in nuclear fusion processes occurring in the interior of stars when there is a diversion from hydrostatic equilibrium. The extension of non-resonant reaction rate integrals in standard, cut-off, depletion, and screening instances is examined in this study, utilising the MacDonald distribution. This study is primarily on the use of generalized special functions that represent the extended non-resonant thermonuclear functions in mathematical physics. Moreover, an attempt is made to analyse the reaction probability integrals corresponding to the several forms of the slowly varying cross-section factor S(E).}



\maketitle

\section{Introduction and preliminaries}
Every facet of the chemical development of the cosmos, or at least its contents, stars and galaxies, is known to be governed by nuclear reactions (\cite{Bethe}). As the mobility of the nuclei is believed to be in thermal equilibrium, \cite{Atk-Hou} postulated that the thermonuclear reactions occurring close to the cores of stars are the source of the energy produced by them. Because of the Coulomb repulsion between nuclei, only the lightest chemical elements could potentially contribute to the state of hot star plasmas. The average of the Gamow penetration factor discovered by \cite{Gam-Tell}, throughout the Maxwell-Boltzmann velocity distribution, determines the energy generation rate (\cite{Fowler}).

Hot stellar fusion plasmas are typically explained using Boltzmann-Gibbs statistical mechanics studied by \cite{SMH}, which is based on the entropy $S_{BG} = -k\, \sum_{i=1}^{W} \, p_{i} \, \ln{p_{i}}$.  The pace at which the chemical composition of hot plasmas changes is described by the coefficient known as the thermonuclear reaction rate. \cite{Critch}, \cite{anderson1994astrophysical}, \cite{HM}, and \cite{Chau-Hau-Mat} all conducted analytical studies of the closed-form representations of thermonuclear reaction rates. It is demonstrated that the reaction rate comprises the astrophysical cross-section factor and its derivatives, which need to be ascertained through experimentation, as well as a significant portion of the thermonuclear reaction rate that is unrelated to experimental outcomes and can be handled using closed-form representation methods involving generalised hypergeometric functions.

Nuclear and neutrino astrophysics are dynamic domains where several experiments gather data to investigate how energy is produced in stars, especially the Sun, based on neutrino observations (\cite{Davis}). Refining analytical representations for nuclear reaction rates is achievable by increasing theoretical understanding and validation through experiments with nuclear cross-sections. The thermonuclear function hypothesis is now being developed further by \cite{Ber-Hall-San}. Thermonuclear fusion reactions are fundamentally responsible for the energy emitted by stellar objects like the Sun.

Assuming that the nuclear velocity distribution is Maxwell-Boltzmannian, collision theory calculates the number of intense collisions that produce nuclear reactions studied by \cite{Fowler}. The general notion of a reaction network is an association of species described by kinetic equations and coupled by a specific set of reactions. Given a certain kind of reaction between two atomic nuclei, $m$ and $n$, the number of reactions per unit of time and volume is equal to the product of the reaction probability ${\langle \sigma v \rangle}$ between the particles and the particle densities $N_{m} N_{n}$, which is defined by
\begin{equation}\label{r1}
  r_{m n} = (1 - \frac{1}{2} \delta_{m n}) N_{m} N_{n} {\langle \sigma v \rangle},
\end{equation}
where $\delta_{m n}$ is the Kronecker symbol and ${\langle \sigma v \rangle}$ denotes the reaction probability, which is the average of the reaction cross section across the normalised distribution function of the particle velocities as follows:
\begin{equation}\label{r2}
  {\langle \sigma v \rangle} = \int_{0}^{\infty} d^{3}v \, f(v) \, \sigma(v) v = \int_{0}^{\infty} dE \sigma(E) \left(\frac{2E}{M}\right)^{1/2} f(E), \quad [{\langle \sigma v \rangle}] = cm^{3} s^{-1}.
\end{equation}
In this instance, $d^{3} v = 4 \pi v^{2} dv$; the kinetic energy of the particles in the center-of-mass system is $E = Mv^{2}/2$, the decreased mass of the particles is $M$, and the reaction cross section is $\sigma(v)$ and $\sigma(E)$, respectively, defined by \cite{MH1}. A geometrical component, $\pi \lambda^{2} \propto E^{-1}$, is always proportional to the quantum-mechanical interaction between two particles, where $\lambda$ is the de Broglie wavelength. Therefore, for nonresonant nuclear reactions involving nuclei with charges $Z_{m}$ and $Z_{n}$ at low energies (i.e., below the Coulomb barrier) (see \cite{anderson1994astrophysical, Chau-Hau-Mat, Fowler, Hau-Joh, HM}), the nonresonant nuclear cross section (in the center-of-mass frame) takes the following form:
\begin{equation}\label{r3}
  \sigma(E) = \frac{S(E)}{E} e^{-2 \pi \eta(E)}, \qquad \qquad \eta(E) = \left(\frac{M}{2}\right)^{\frac{1}{2}} \frac{Z_{m} Z_{n} \, q^{2}}{\hbar E^{1/2}},
\end{equation}
where $e^{-2 \pi \eta(E)}$ is the Gamow penetration factor with the Sommerfeld parameter $\eta(E)$, $q$ is the quantum electric charge, and $\hbar$ is Planck's quantum of action. Equation \eqref{r3}, which incorporates the S-factor studied by \cite{Fowler}, makes it easier to convert observed reaction cross-sections into astrophysical energies. A slowly fluctuating function of the center-of-momentum energy defined by \cite{HM} and inherently nuclear components of the probability of a nuclear reaction studied by \cite{Chau-Hau-Mat}, are indicated by the cross-section factor $S(E)$ in \eqref{r3}, a residual energy function. It may be expressed conveniently as a power series expansion, where
\begin{equation}\label{r4}
  S(E) \approx S(0) + S^{\prime} (0) E + \frac{1}{2} S^{\prime\prime} (0) E^{2} = \sum_{\upsilon = 0}^{2} \frac{S^{(\upsilon)} (0)}{\upsilon !} E^{\upsilon}.
\end{equation}
The value of $S(E)$ at zero energy is represented by $S(0)$, whereas the first and second derivatives of $S(E)$ with respect to energy evaluated at $E = 0,$ are $S^{\prime}(0)$ and $S^{\prime\prime}(0)$, respectively.

This paper proposes the introduction of a new non-resonant thermonuclear reaction rate of the particles in the center-of-mass system by making use of the relative velocity and kinetic energy distribution function that involves the Bessel function of the third kind (i.e, MacDonald function). We then study the thermonuclear collision probability integral, called the astrophysical thermonuclear functions, and extend by using the MacDonald function. Further, we evaluate the Mellin transform and Fox $H$-function illustration of the extended thermonuclear functions, which present the closed-form representation of the extended thermonuclear rate of reaction.

\section{Non-Resonant Thermonuclear Rate of Reactions}
In this section, we investigated the Maxwell-Boltzmannian and Mittag-Leffler cases of relative velocity and kinetic energy distribution function. Following this, we studied the Maxwell-Boltzmannian and Mittag-Leffler rate of reaction. Further, we proposed the MacDonald rate of reaction by making use of the MacDonald case of relative velocity and kinetic energy distribution function.

\subsection{Maxwell-Boltzmannian Thermonuclear Reaction}
In a non-degenerate and non-relativistic gas, the relative velocities of the nuclei have a Maxwell-Boltzmannian distribution function (\cite{Hau-Joh}), where
\begin{equation}\label{r5}
  f_{MBD}(v) dv = \left(\frac{M}{2 \pi K T} \right)^{\frac{3}{2}} \, e^{-\frac{M v^{2}}{2 K T}} 4 \pi v^{2} \, dv,
\end{equation}
where $T$ is the absolute temperature and $K$ is the Boltzmann constant. The Maxwell-Boltzmannian relative kinetic energy spectrum for \eqref{r3} through considering a non-degenerate, non-relativistic gas of particles (\cite{MH1}) is
\begin{equation}\label{r6}
  f_{MBD} (E) dE = 2\pi \left(\frac{1}{\pi K T}\right)^{\frac{3}{2}} \, e^{-\frac{E}{K T}} \sqrt{E} dE,
\end{equation}
where $E = \frac{M v^{2}}{2}.$ If the distribution function $\varepsilon f (v, t)$ describes the state of the plasma at time $t,$ where $\varepsilon$ is the constant particle number density, $V$ is the velocity variable, and $v = |V|$, then the conservation of mass and energy imply that
\begin{align*}
  &\int d^{3}v f(v, t) = 1  \\
  & \int d^{3}v v^{2} f(v, t) = \frac{3 K T}{M}.
\end{align*}
As $t \rightarrow \infty, f(v, t)$ tends to the Maxwell-Boltzmann distribution function given in \eqref{r5}, this is the case in a gravitationally stabilised star fusion reactor. Using the kinetic energy distribution function from \eqref{r6} in \eqref{r2}, the Maxwell-Boltzmannian reaction probability we have
\begin{equation}\label{r7}
  {\langle \sigma v \rangle} = \left(\frac{8}{\pi M}\right)^{\frac{1}{2}} \left(\frac{1}{K T}\right)^{\frac{3}{2}} \int_{0}^{\infty} E \sigma(E) e^{-\frac{E}{K T}} dE.
\end{equation}

Substituting \eqref{r3} following \eqref{r4} in \eqref{r7} and then using that in \eqref{r1}, the non-resonant Maxwell-Boltzmannian thermonuclear rate of reaction we have
\begin{align}\label{r8}
& r_{m n} = (1 - \frac{1}{2} \delta_{m n}) N_{m} N_{n} \left(\frac{8}{\pi M}\right)^{\frac{1}{2}} \left(\frac{1}{K T}\right)^{\frac{3}{2}} \notag\\
&\qquad \qquad \qquad \times \sum_{\upsilon = 0}^{2} \frac{S^{(\upsilon)} (0)}{\upsilon !} \left\{\int_{0}^{\infty} E^{\upsilon} e^{-\frac{E}{K T}} e^{-\frac{\tau}{\sqrt{E}}} \, dE\right\},
\end{align}
where
\begin{equation}\label{r9}
\tau = 2\pi \left(\frac{M}{2}\right)^{\frac{1}{2}} \frac{Z_{m} Z_{n} \, q^{2}}{\hbar}.
\end{equation}

\subsection{Mittag-Leffler Thermonuclear Reaction}
Recently, \cite{Hau-Kab-Kum} by making use of the two-parameter Mittag-Leffler function, presented the Mittag-Leffler velocity distribution as
\begin{equation}\label{r10}
  f_{ML} (v) dv = C\, {\mathbb{E}}_{\beta, \beta} \left(-\frac{M v^{2}}{2 K T}\right) 4 \pi v^{2} dv \quad\qquad 0 < \beta \leq 1,
\end{equation}
where $C$ is the normalizing constant and ${\mathbb{E}}_{\beta, \beta} (\cdot)$ is the Mittag-Leffler function (see, \cite{Gor-Kil-Main-Rog, DLMF}). The Mittag-Leffler energy distribution, which measures relative kinetic energy, is
\begin{equation}\label{r11}
  f_{ML} (E) dE = \frac{2 \sqrt{\pi}}{\beta} \Gamma\left(1 - \frac{\beta}{2}\right) \left(\frac{1}{\pi K T}\right)^{\frac{3}{2}} {\mathbb{E}}_{\beta, \beta} \left(-\frac{E}{K T}\right) \sqrt{E} dE,
\end{equation}
where $0 < \beta \leq 1.$ Then, using \eqref{r11} in \eqref{r2} followed by \eqref{r1}, the modified form, called the Mittag-Leffler reaction rate probability integral, we have
\begin{align}\label{r12}
^{ML}\bar{r}_{m n} & = (1 - \frac{1}{2} \delta_{m n}) N_{m} N_{n} \left(\frac{8 \pi}{M}\right)^{\frac{1}{2}} \frac{\Gamma\left(1 - \frac{\beta}{2}\right)}{\beta} \left(\frac{1}{\pi K T}\right)^{\frac{3}{2}} \notag\\ &\qquad \qquad  \times \sum_{\upsilon = 0}^{2} \frac{S^{(\upsilon)} (0)}{\upsilon !} \left\{\int_{0}^{\infty} E^{\upsilon} {\mathbb{E}}\left({-\frac{E}{K T}}\right) e^{-\frac{\tau}{\sqrt{E}}} \, dE\right\},
\end{align}
where $\tau$ is given in \eqref{r9}.

\subsection{MacDonald Thermonuclear Reaction}
The MacDonald velocity distribution function, which we proposed as
\begin{equation}\label{m1}
  f_{MD}(v) dv = \left(\frac{2}{\pi}\right)^{\frac{1}{2}} \Gamma\left(\frac{1 - \alpha}{2}\right) \left(\frac{M}{2 \pi K T}\right)^{2} \, \mathbb{K}_{\alpha + \frac{1}{2}} \left(\frac{M v^{2}}{2 K T}\right) 4 \pi v^{3} dv,
\end{equation}
\begin{equation*}
  (0 < \alpha \leq 1).
\end{equation*}
Considering the relative kinetic energy, $E = \frac{M v^{2}}{2}$ in the above expression, the MacDonald energy distribution function, we obtain
\begin{equation}\label{m2}
  f_{MD}(E) dE = 2\sqrt{2 \pi}\, \Gamma\left(\frac{1 - \alpha}{2}\right) \left(\frac{1}{\pi K T}\right)^{2} \, \mathbb{K}_{\alpha + \frac{1}{2}} \left(\frac{E}{K T}\right)\, E\,dE,
\end{equation}
where $0 < \alpha \leq 1$ and $\mathbb{K}_{\alpha} (\cdot)$ is the MacDonald function (\cite{Wats}) (i.e, Bessel function of third kind of order $\alpha$).
As \begin{equation*}
  \mathbb{K}_{\frac{1}{2}} (x) = \sqrt{\frac{\pi}{2 x}} e^{-x},
\end{equation*}
then for $\alpha = 0$ in \eqref{m1} and \eqref{m2} reduce to the Maxwell-Boltzmannian relative distribution function given in \eqref{r5} and \eqref{r6}, respectively. Now substituting \eqref{m2} in \eqref{r2} in accordance with \eqref{r3} and \eqref{r4} and then putting in \eqref{r1}, the rate reaction probability integral we obtain
\begin{align}\label{m3}
^{MD}\hat{r}_{m n} & = (1 - \frac{1}{2} \delta_{m n}) N_{m} N_{n} \left(\frac{2}{M}\right)^{\frac{1}{2}} \left(\frac{2}{\pi}\right)^{\frac{3}{2}} \Gamma\left(\frac{1 - \alpha}{2}\right) \left(\frac{1}{K T}\right)^{2} \notag\\ &\qquad \qquad  \times \sum_{\upsilon = 0}^{2} \frac{S^{(\upsilon)} (0)}{\upsilon !} \left\{\int_{0}^{\infty} E^{\upsilon + \frac{1}{2}}\, \mathbb{K}_{\alpha + \frac{1}{2}} \left(\frac{E}{K T}\right)\, e^{-\frac{\tau}{\sqrt{E}}} \, dE\right\},
\end{align}
where $\tau$ is given in \eqref{r9}. The transformation of $y = \frac{E}{K T}$ and the dimensionless positive real variable in physical case, $x = \frac{\tau}{\sqrt{K T}}$ in \eqref{m3} yields
\begin{align}\label{m4}
  ^{MD}\hat{r}_{m n} & = (1 - \frac{1}{2} \delta_{m n}) N_{m} N_{n} \left(\frac{2}{M}\right)^{\frac{1}{2}} \left(\frac{2}{\pi}\right)^{\frac{3}{2}} \Gamma\left(\frac{1 - \alpha}{2}\right) \notag\\ &\qquad \quad  \times \sum_{\upsilon = 0}^{2} \left(\frac{1}{K T}\right)^{-\upsilon + \frac{1}{2}} \frac{S^{(\upsilon)} (0)}{\upsilon !} \left\{\int_{0}^{\infty} y^{\upsilon + \frac{1}{2}}\, \mathbb{K}_{\alpha + \frac{1}{2}}(y)\, e^{-x y^{\frac{-1}{2}}} \, dE\right\}.
\end{align}
The modified form of the thermonuclear reaction defined in \eqref{m4} shall be known as \textit{MacDonald extended thermonuclear rate of reaction} (MacDonald rate reaction probability integral).

Here we present the graphs of the MacDonald energy distribution function in \eqref{m2}. Figure~\ref{fig.1} illustrates the graph of $f_{MD}(E)$ for various $KT$ values at a fixed value of $\alpha = 0.20,$ while Figure~\ref{fig.2} illustrates the graph of $f_{MD}(E)$ for various $\alpha$ values, with fixed values of $KT = 400.$ Figure~\ref{fig.3} highlights the $3$D graph of the MacDonald energy distribution for a fixed value of $\alpha = 0.20$ through $KT$ values ranging from $200$ to $400$, and Figure~\ref{fig.4} demonstrates the $3$D graph of the MacDonald energy distribution for various values of $\alpha$ ranging from $0.04$ to $0.36,$ with $KT = 400.$
\begin{figure}[htbp]
  \centering
  % Left Independent Figure
  \begin{minipage}[b]{0.45\textwidth}
    \centering
    \includegraphics[width=\textwidth]{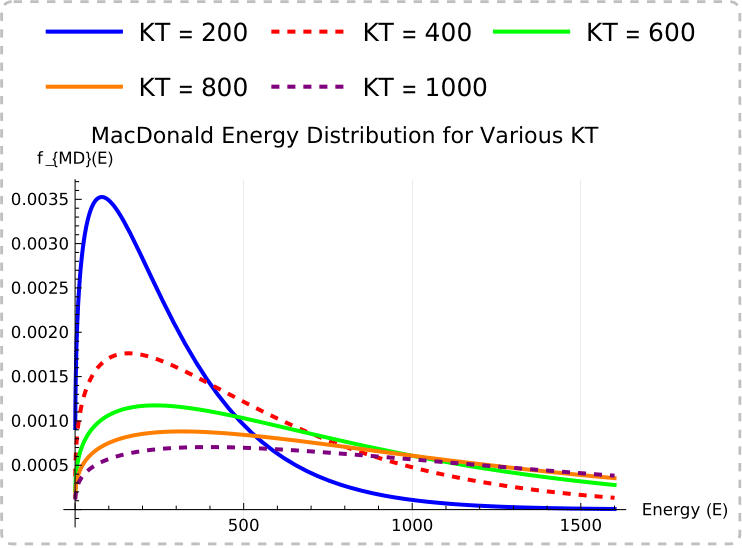}
    \caption{$f_{MD}(E)$  for $\alpha$ = 0.20 and KT = 200, 400, 600, 800, 1000.}
    \label{fig.1}
  \end{minipage}
  \hfill
  % Right Independent Figure
  \begin{minipage}[b]{0.45\textwidth}
    \centering
    \includegraphics[width=\textwidth]{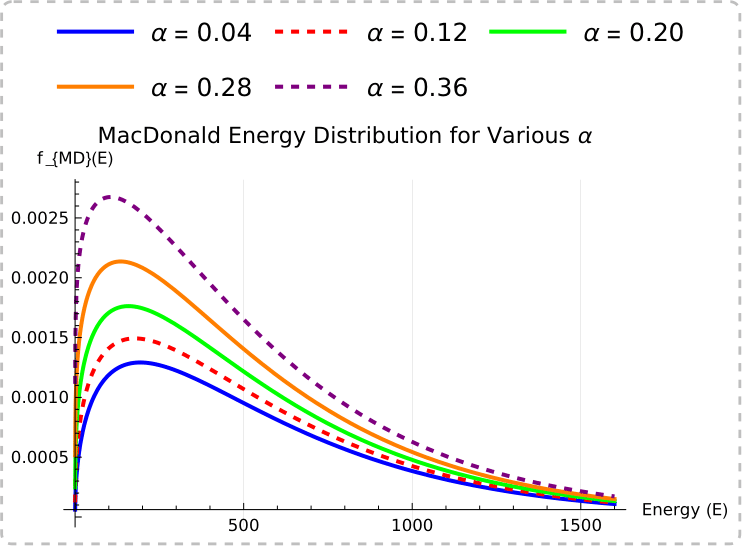}
    \caption{$f_{MD}(E)$  for KT = 400 and $\alpha$ = 0.04, 0.12, 0.20, 0.28, 0.36.}
    \label{fig.2}
  \end{minipage}
\end{figure}
\begin{figure}[htbp]
  \centering
  % Left Independent Figure
  \begin{minipage}[b]{0.45\textwidth}
    \centering
    \includegraphics[width=\textwidth]{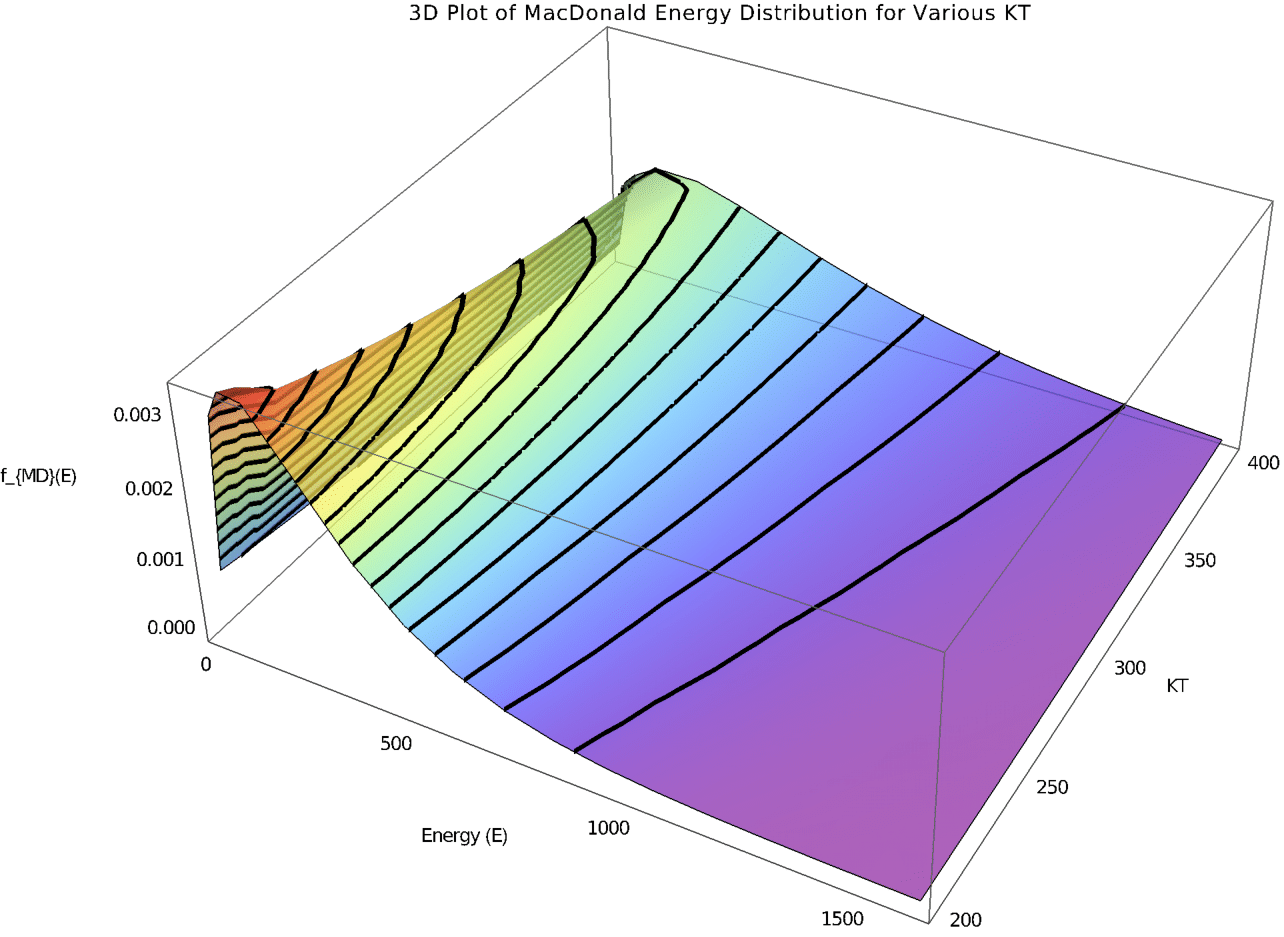}
    \caption{$f_{MD}(E)$ ($3$D) for a fixed $\alpha$ = 0.20 and KT = \{200, \dots, 400\}}
    \label{fig.3}
  \end{minipage}
  \hfill
  % Right Independent Figure
  \begin{minipage}[b]{0.45\textwidth}
    \centering
    \includegraphics[width=\textwidth]{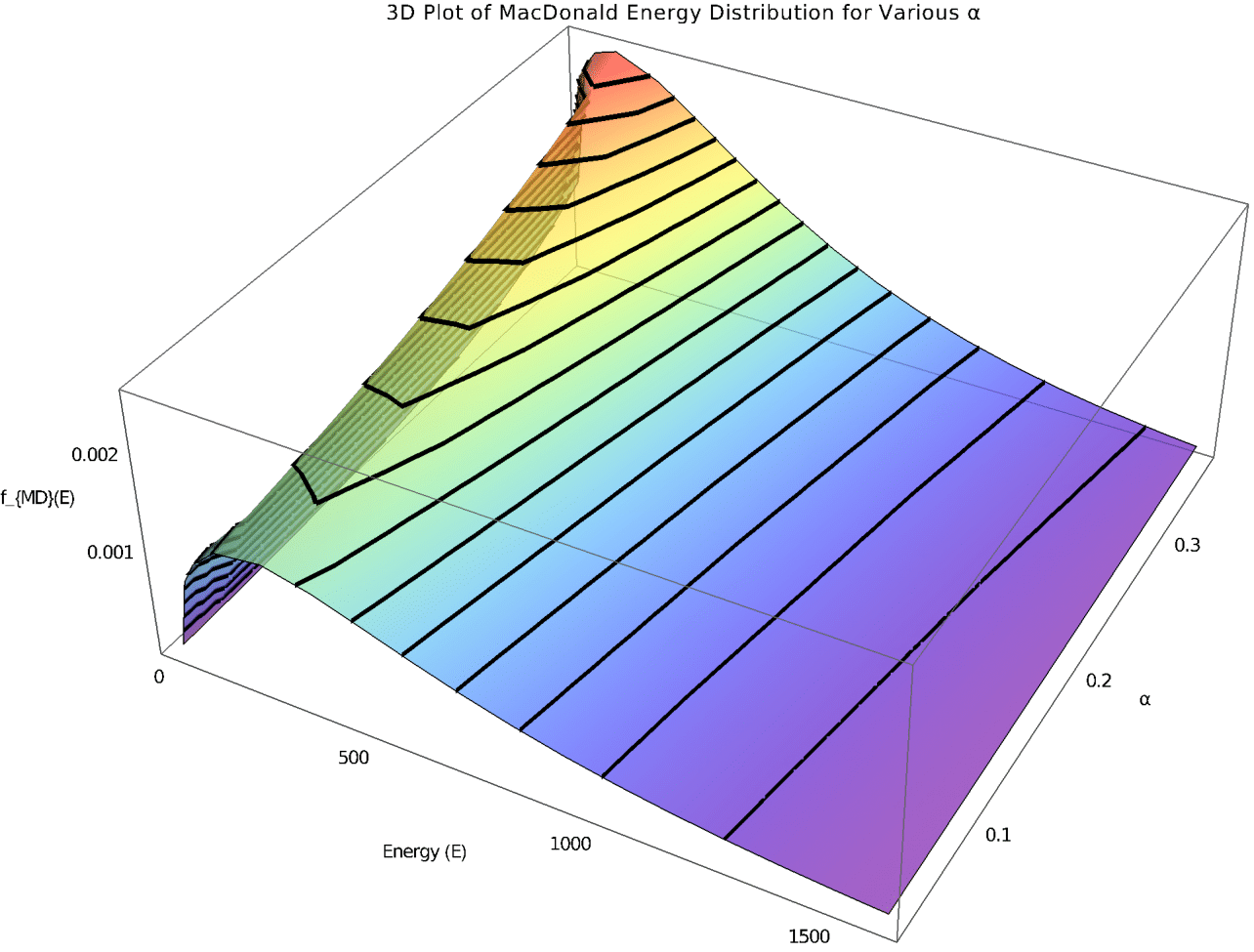}
    \caption{$f_{MD}(E)$ ($3$D)  for a fixed KT = 400 and $\alpha$ = \{0.04, \dots, 0.36\}.}
    \label{fig.4}
  \end{minipage}
\end{figure}

\newpage

\section{The Astrophysical Thermonuclear Functions}

As an isotropic nuclear reaction consisting of the cross section $\sigma$ independent of the collision angle, in 1994, \cite{anderson1994astrophysical} defined the standard case of the non-resonant thermonuclear function as
\begin{equation}\label{t1}
  \mathcal{I}_{1} (x, \upsilon) = \int_{0}^{\infty} y^{\upsilon} e^{-y} e^{-xy^{-\frac{1}{2}}} dy,
\end{equation}
where $y = \frac{E}{K T}$ and $x = \frac{\tau}{\sqrt{K T}}$ ($\tau$ is given in \eqref{r9}). Since the high-energy tail of the Maxwell-Boltzmann distribution may be cut off by dissipative collision events in the thermonuclear plasma, they stated
\begin{equation}\label{t2}
  \mathcal{I}_{2}^{(d)} (x, \upsilon) = \int_{0}^{d} y^{\upsilon} e^{-y} e^{-xy^{-\frac{1}{2}}} dy
\end{equation}
where $d$ denotes a certain typically high energy and for \eqref{t1}.
\begin{equation}\label{t3}
  \mathcal{I}_{3} (x, t, \upsilon) = \int_{0}^{\infty} y^{\upsilon} e^{-y} e^{-x(y+t)^{-\frac{1}{2}}} dy,
\end{equation} is the conventional Maxwell-Boltzmannian thermonuclear function that accounts for screening effects,
where $t$ is the electron screening parameter. Lastly, the thermonuclear function can be expressed as follows:
\begin{equation}\label{t4}
  \mathcal{I}_{4} (x, \delta, b, \upsilon) = \int_{0}^{\infty} y^{\upsilon} e^{-y} e^{-by^{\delta}} e^{-xy^{-\frac{1}{2}}} dy,
\end{equation}
where the high-energy tail of the Maxwell-Boltzmann distribution is enhanced or reduced by the parameter $\delta$.

While establishing astrophysical thermonuclear functions for Boltzmann-Gibbs statistics and Tsallis statistics, \cite{SMH} defined the generalization of standard, cut-off, depleted, and screened cases of non-resonant thermonuclear functions as
\begin{equation}\label{t5}
  \mathcal{I}_{1} (\upsilon-1, a, x, \varpi) = \int_{0}^{\infty} y^{\upsilon-1} e^{-ay} e^{-xy^{\varpi}} dy,
\end{equation}
\begin{equation*}
  (\Re(\upsilon)>0, \Re(a)>0, \Re(x)>0, \Re(\varpi)>0),
\end{equation*}
\begin{equation}\label{t6}
  \mathcal{I}_{2}^{(d)} (\upsilon-1, a, x, \varpi) = \int_{0}^{d} y^{\upsilon-1} e^{-ay} e^{-xy^{\varpi}} dy,
\end{equation}
\begin{equation*}
  (\Re(\upsilon)>0, \Re(a)>0, \Re(x)>0, \Re(\varpi)>0),
\end{equation*}
\begin{equation}\label{t7}
  \mathcal{I}_{3} (\upsilon-1, a, b, \delta, x, \varpi) = \int_{0}^{\infty} y^{\upsilon-1} e^{-ay-by^{\delta}} e^{-xy^{\varpi}} dy,
\end{equation}
\begin{equation*}
  (\Re(\upsilon)>0, \Re(a)>0, \Re(x)>0, \Re(\varpi)>0, \Re(b)>0, \delta>0),
\end{equation*}
\begin{equation}\label{t8}
  \mathcal{I}_{4} (\upsilon-1, a, x, t, \varpi) = \int_{0}^{\infty} y^{\upsilon-1} e^{-ay} e^{-x(y+t)^{\varpi}} dy,
\end{equation}
\begin{equation*}
  (\Re(\upsilon)>0, \Re(a)>0, \Re(x)>0, \Re(\varpi)>0),
\end{equation*}
respectively. In 2023, for $z>0, x>0, \varpi > 0$ and $0<\beta\leq 1$, \cite{Hau-Kab-Kum} evaluated a set of four general thermonuclear functions for the two-parameter Mittag-Leffler case, considering the standard, cut-off, depleted, and screening instances, which are of the forms,
\begin{align}\label{t9}
&\mathcal{I}_{1, \beta} (\upsilon, z, x, \varpi) = \int_{0}^{\infty} y^{\upsilon - 1} \mathbb{E}_{\beta, \beta} (-zy) e^{-xy^{-\varpi}}\, dy,\\
&\mathcal{I}_{2, \beta}^{(d)} (\upsilon, z, x, \varpi) = \int_{0}^{d} y^{\upsilon - 1} \mathbb{E}_{\beta, \beta} (-zy) e^{-xy^{-\varpi}}\, dy,\qquad d<\infty,\\
&\mathcal{I}_{3, \beta} (\upsilon, t, z, \delta, x, \varpi) = \int_{0}^{\infty} y^{\upsilon - 1} \mathbb{E}_{\beta, \beta} (-ty) e^{-zy^{\delta}-xy^{-\varpi}}\, dy, \qquad \delta>0,\\
&\mathcal{I}_{4, \beta} (\upsilon, z, x, u, \varpi) = \int_{0}^{\infty} y^{\upsilon - 1} \mathbb{E}_{\beta, \beta} (-zy) e^{-x(u+y)^{-\varpi}}\, dy, \qquad u>0,
\end{align}
respectively.

For $\Re(\upsilon)>0, x>0$ and $0<\alpha \leq 1$, the integral appeared in \eqref{m4} to be evaluated in this case is of the form
\begin{equation}\label{t10}
  \mathcal{T}_{1, \alpha}(\upsilon, x) = \int_{0}^{\infty} y^{\upsilon + \frac{1}{2}} \mathbb{K}_{\alpha + \frac{1}{2}} (y)\, e^{-xy^{-\frac{1}{2}}} dy,
\end{equation}
That is known as the collision probability integral, called the astrophysical thermonuclear function for non-resonant thermonuclear reactions in the MacDonald case. The more general thermonuclear function for the standard non-resonant case, we consider
\begin{equation}\label{t11}
\mathcal{T}_{1, \alpha}(\upsilon, z, x, \gamma) \stackrel{\text{def}}{=} \int_{0}^{\infty} y^{\upsilon - 1} \mathbb{K}_{\alpha + \frac{1}{2}} (zy)\, e^{-xy^{-\gamma}} dy,
\end{equation} where $z>0, x>0, \Re(\upsilon)>0, \Re(\gamma)>0$ and $0<\alpha\leq 1.$ The MacDonald extended thermonuclear functions in cut-off, depleted, and screened cases are of the following, we present
\begin{align}\label{t12}
\mathcal{T}_{2, \alpha}^{(d)} (\upsilon, z, x, \gamma)& \stackrel{\text{def}}{=} \int_{0}^{d} y^{\upsilon - 1} \mathbb{K}_{\alpha + \frac{1}{2}} (zy)\, e^{-xy^{-\gamma}} dy, \qquad d<\infty, \\
\mathcal{T}_{3, \alpha}(\upsilon, t, z, \delta, x, \gamma) & \stackrel{\text{def}}{=} \int_{0}^{\infty} y^{\upsilon - 1} \mathbb{K}_{\alpha + \frac{1}{2}} (ty)\, e^{-zy^{\delta}-xy^{-\gamma}} dy, \qquad \delta>0, \\
\mathcal{T}_{4, \alpha}(\upsilon, z, x, l, \gamma) & \stackrel{\text{def}}{=} \int_{0}^{\infty} y^{\upsilon - 1} \mathbb{K}_{\alpha + \frac{1}{2}} (zy)\, e^{-x(l+y)^{-\gamma}} dy, \qquad l>0,
\end{align}
respectively, where $z>0, x>0, \Re(\upsilon)>0, \Re(\gamma)>0$ and $0<\alpha\leq 1.$

\section{Closed-Form Representations}
In his study on symmetrical Fourier kernels, C. Fox defines the $H$-function as the Mellin-Barnes type path integral (see \cite{KS,MSH}):
\begin{align}\label{H-function} %%%(3.3)
&H^{a,\;b}_{\phi,\;\psi}(z)=H^{a,\;b}_{\phi,\;\psi}\left[z\Biggm| \begin{array}{c}
(\zeta_{i},P_{i})_{1,\phi}       \\
(\eta_{j},Q_{j})_{1,\psi}    \\
\end{array} \right]=H^{a,\;b}_{\phi,\;\psi}\left[z\Biggm| \begin{array}{c}
 (\zeta_{1},P_{1}),\cdots,(\zeta_{\phi},P_{\phi})          \\
  (\eta_{1},Q_{1}),\cdots,(\eta_{\psi},Q_{\psi})  \\
\end{array} \right]\notag\\
&\qquad\qquad\hskip 15mm=\frac{1}{2\pi i}\int_{\mathcal{L}}\;\Theta(s)\;z^{-s}\;ds\,,
\end{align}
where
\begin{equation}\label{H-function-Theta}
\Theta(s)=\frac{\prod\limits_{j=1}^{a}\Gamma(\eta_{j}+Q_{j}s)\prod\limits_{j=1}^{b}\Gamma(1-\zeta_{j}-P_{j}s)}
{\prod\limits_{j=a+1}^{\psi}\Gamma(1-\eta_{j}-Q_{j}s)\prod\limits_{j=b+1}^{\phi}\Gamma(\zeta_{j}+P_{j}s)}\,,
\end{equation}
and a good Mellin-Barnes type contour that starts at $\epsilon-i\infty$ and finishes at $\epsilon+i\infty\,(\epsilon \in \mathbb{R})$ is $\mathcal{L}=\mathcal{L}_{(i\epsilon;\infty)}$, such that every pole of $\Gamma(\eta_{j}+Q_{j}s)\;( j=1,\cdots,a)$ is different from every pole of $\Gamma(1-\zeta_{j}-P_{j}s)\;( j=1,\cdots,b)$. When an empty product is read as $1$, the inequalities are satisfied by the integers $a,\;b,\;\phi,\;\psi$. $0\leq a\leq \;\psi$ and $0\leq b\leq \;\phi$, the complex parameters $\zeta_{j}\;( j = 1,\cdots,\phi)$ and $\eta_{j}\;( j = 1,\cdots,\psi)$ are so limited that the integrand's poles do not coincide, and the coefficients $P_{j}\;(j=1,\cdots,\phi)$ and $Q_{j}\;(j=1,\cdots,\psi)$ are positive real numbers.

The two-variable $H$-function is defined by
\begin{align}\label{H-Function-2}
&H[y,z]=H^{\;\mu_{1},0:\mu_{2},\nu_{2};\mu_{3},\nu_{3}}_{\;\chi_{1},\omega_{1}:\chi_{2},\omega_{2};\chi_{3},\omega_{3}}\left[\begin{array}{c}y\\z \end{array}\,\Biggm| \begin{array}{c}
 (a_{l};\alpha_{l},\mathcal{A}_{l})_{1,\chi_{1}}:(c_{l},\mathcal{C}_{l})_{1,\chi_{2}};(e_{l},\mathcal{E}_{l})_{1,\chi_{3}}          \\
 (b_{l};\beta_{l},\mathcal{B}_{l})_{1,\omega_{1}}:(d_{l},\mathcal{D}_{l})_{1,\omega_{2}};(f_{l},\mathcal{F}_{l})_{1,\omega_{3}}    \\
\end{array} \right]\notag\\
&\hskip 5mm=\frac{1}{(2\pi i)^{2}}\int_{\mathcal{O}_{s}}\int_{\mathcal{O}_{t}}\; \theta_{1}(s)\;\theta_{2}(t)\;\Delta(s,\;t)\;y^{-s}\;z^{-t}\;ds\;dt
\end{align}
such that $y\neq0 \neq z$, and $i=\sqrt{-1}$, where $\theta_{1}(s)$, $\theta_{2}(t)$, and $\Delta(s,\;t)$ are provided by
$$\Delta(s,\;t)=\frac{\prod\limits_{l=1}^{\mu_{1}}\Gamma(b_{l}+\beta_{l}s+\mathcal{B}_{l}t)}{\prod\limits_{l=\mu_{1}+1}^{\omega_{1}}
\Gamma(1 - b_{l}-\beta_{l}s-\mathcal{B}_{l}t)\prod\limits_{l=1}^{\chi_{1}}\Gamma(a_{l}-\alpha_{l}s+\mathcal{B}_{l}t)}$$
$$\theta_{1}(s)=\frac{\prod\limits_{l=1}^{\nu_{2}}\Gamma(1-c_{l}-\mathcal{C}_{l}s)
\prod\limits_{l=1}^{\mu_{2}}\Gamma(d_{l}+\mathcal{D}_{l}s)}{\prod\limits_{l=\nu_{2}+1}^{\chi_{2}}\Gamma(c_{l}+\mathcal{C}_{l}s)
\prod\limits_{l=\mu_{2}+1}^{\omega_{2}}\Gamma(1-d_{l}-\mathcal{D}_{l}s)}$$
$$\theta_{2}(t)=\frac{\prod\limits_{l=1}^{\nu_{3}}\Gamma(1-e_{l}-\mathcal{E}_{l}t)
\prod\limits_{l=1}^{\mu_{3}}\Gamma(f_{l}+\mathcal{F}_{l}t)}{\prod\limits_{l=\nu_{3}+1}^{\chi_{3}}\Gamma(e_{l}+\mathcal{E}_{l}t)
\prod\limits_{l=\mu_{3}+1}^{\omega_{3}}\Gamma(1-f_{l}-\mathcal{F}_{l}t)}$$
where $\mathcal{O}_{s}$ is a suitable Mellin-Barnes type contour that begins at the point $-i\infty$ and ends at the point $+i\infty$ so that all of the poles of $\Gamma(d_{l}+\mathcal{D}_{l}s)\;( l=1,\cdots,\mu_2)$ are distinct from those of $\Gamma(1-c_{l}-\mathcal{C}_{l}s)\;( l=1,\cdots,\nu_2)$ and $\Gamma(b_{l}+\beta_{l}s+\mathcal{B}_{l}t)\;( l=1,\cdots,\mu_1)$.
Additionally, $\mathcal{O}_{t}$ represents a suited contour of the Mellin-Barnes type that begins at the point $-i\infty$ and ends at the point $+i\infty$ so that all of the poles of $\Gamma(f_{l}+\mathcal{F}_{l}t)\;( l=1,\cdots,\mu_3)$ are distinct from those of $\Gamma(1-e_{l}-\mathcal{E}_{l}t)\;( l=1,\cdots,\nu_3)$ and $\Gamma(b_{l}+\beta_{l}s+\mathcal{B}_{l}t)\;( l=1,\cdots,\mu_1)$. An empty product is regarded as 1, the non-negative integers $\mu_{l},\;\nu_{l},\;\chi_{l},\;\omega_{l}$ comply with the inequalities $0\leq \mu_{l}\leq \;\omega_{l}$, $\omega_{l}\geq0$ and $0\leq \nu_{l}\leq \;\chi_{l}$\;(l=1,2,3;\,l=1,2), the coefficient $\alpha's$, $\beta's$, $ \mathcal{A}'s$, $\mathcal{B}'s$, $\mathcal{C}'s$, $\mathcal{D}'s$, $\mathcal{E}'s$ and the complex parameters $a_{l}\;( l = 1,\cdots,\chi)$ and $b_{l}\;( l = 1,\cdots,\omega)$ are so restricted that none of the integrand's poles coincide, and $\mathcal{F}'s$ are positive real numbers.
Refer to \cite{KS, MSH} for further information on Fox's $H$-function in one variable and also in two variables. The $H$-function in \eqref{H-function} reduces to the Meijer $G$-function (\cite{KS,MSH}) if $P_{j}=Q_{j}=1$.
\begin{align}\label{GF} %%%(3.3)
&G^{a,\;b}_{\phi,\;\psi}(z)=G^{a,\;b}_{\phi,\;\psi}\left[z\Biggm| \begin{array}{c}
 \zeta_{1},\cdots,\zeta_{\phi} \\
 \eta_{1},\cdots,\eta_{\psi}  \\
\end{array} \right]=\frac{1}{2\pi i}\int_{\mathcal{L}}\;\Lambda(s)\;z^{-s}\;ds\,,
\end{align}
where
\begin{equation}\label{G-function-Lambda}
\Lambda(s)=\frac{\prod\limits_{j=1}^{a}\Gamma(\eta_{j}+s)\prod\limits_{j=1}^{b}\Gamma(1-\zeta_{j}-s)}
{\prod\limits_{j=a+1}^{\psi}\Gamma(1-\eta_{j}-s)\prod\limits_{j=b+1}^{\phi}\Gamma(\zeta_{j}+s)}\,.
\end{equation}

\cite{Hau-Joh} evaluated the collision probability integral between the particles (i.e, astrophysical thermonuclear function) given in \eqref{t1} and found the closed-form representation in terms of Meijer $G$-function as
\begin{equation}\label{tg1}
  \mathcal{I}_{1} (x, \upsilon) = \pi^{-\frac{1}{2}}\,G^{3,\;0}_{0,\;3}\left[\left(\frac{x}{2}\right)^{2}\Biggm| \begin{array}{c}
                             \\
 1+ \upsilon, \frac{1}{2}, 0  \\
\end{array} \right], \qquad (|\arg x^{2}|<\frac{3}{2} \pi).
\end{equation}
The astrophysical thermonuclear functions in line with the Maxwell-Boltzmann, Boltzmann-Gibbs, and Tsallis distributions were quantitatively examined by several authors (see \cite{anderson1994astrophysical, Chau-Hau-Mat, Hau-Joh, Hau-Kuma, HM, Kum-Hau, SMH}), who subsequently developed closed-form representations for non-resonant standard, cut-off, and depleted tail circumstances. \cite{Hau-Kab-Kum} established the analytic forms of thermonuclear functions, proposing the Mittag-Leffler relative velocity and energy distribution functions. Moreover, they evaluated the closed-form representations of the proposed Mittag-Leffler case of thermonuclear functions.

\subsection{Evaluation of (3.14) and (3.16)}

\begin{thm}\label{Th1}
Let $\Re(\upsilon)>0, \Re(\gamma)>0, z>0, x>0$ and $0<\alpha \leq 1.$ Then the following $H$-function representation is valid:
\begin{equation}\label{H1}
\mathcal{T}_{1, \alpha}(\upsilon, z, x, \gamma) = \frac{1}{4} \left(\frac{2}{z}\right)^{\upsilon} H^{3,\;0}_{0,\;3}\left[x\left(\frac{z}{2}\right)^{\gamma}\Biggm| \begin{array}{c}
   --    \\
(0, 1), \left(\frac{\upsilon + \alpha}{2} + \frac{1}{4}, \frac{\gamma}{2}\right), \left(\frac{\upsilon - \alpha}{2} - \frac{1}{4}, \frac{\gamma}{2}\right)  \\
\end{array} \right].
\end{equation}
\end{thm}
\begin{proof}
By making use of the one-variable Mellin transform (\cite{MSH})
\begin{equation*}
\mathcal{M}_{\mathcal{T}_{1, \alpha}} (s) = \int_{0}^{\infty} x^{s-1} \left\{\int_{0}^{\infty} y^{\upsilon-1} f(x) \, dy\right\} dx,
\end{equation*}
to the standard thermonuclear function \eqref{t11}, we have
\begin{equation*}
\mathcal{M}_{\mathcal{T}_{1, \alpha}} (s) = \int_{0}^{\infty} x^{s-1} \left\{\int_{0}^{\infty} y^{\upsilon-1} \mathbb{K}_{\alpha + \frac{1}{2}} (zy)\, e^{-xy^{-\gamma}} \, dy\right\} dx.
\end{equation*}
By utilising \cite[p.676, Eq.(6.561(16))]{Gr-Ry} and swapping the integrals because of the uniform convergence, we can solve
\begin{equation}\label{S1}
\mathcal{M}_{\mathcal{T}_{1, \alpha}} (s) = \frac{1}{4} \left(\frac{2}{z}\right)^{\upsilon + \gamma s} \Gamma(s) \Gamma\left(\frac{\upsilon + \alpha}{2} + \frac{1}{4} + \frac{\gamma}{2}s\right) \Gamma\left(\frac{\upsilon - \alpha}{2} - \frac{1}{4} + \frac{\gamma}{2}s\right),
\end{equation}
where $\Re(z)>0, \Re(s)>0, \Re(\upsilon+\gamma s \pm \alpha + \frac{1}{2})>0.$ Employing the inverse Mellin transform to \eqref{S1}, the Mellin-Barnes integral we obtain
  \begin{equation}\label{S2}
    \mathcal{T}_{1, \alpha} = \frac{1}{4} \left(\frac{2}{z}\right)^{\upsilon} \frac{1}{2 \pi \iota} \int_{\mathcal{L}} \Gamma(s) \Gamma\left(\frac{\upsilon + \alpha}{2} + \frac{1}{4} + \frac{\gamma}{2}s\right) \Gamma\left(\frac{\upsilon - \alpha}{2} - \frac{1}{4} + \frac{\gamma}{2}s\right) \left(\frac{z^{\gamma}}{2^{\gamma}}\right)^{-s} ds.
  \end{equation}
  Now comparing \eqref{S2} to the definition of $H$-function \eqref{H-function} yields the required result.
\end{proof}

\begin{cor}\label{Cra}
  Let $\Re(\upsilon)>0, x>0$ and if we take the transformation of $\upsilon \mapsto \upsilon + \frac{3}{2}$ and for $z = 1, \gamma = \frac{1}{2},$ then
  \begin{equation}\label{Cr1}
    \mathcal{T}_{1, \alpha} (\upsilon, 1, x, \frac{1}{2}) = 2^{\upsilon - \frac{1}{2}} H^{3,\;0}_{0,\;3}\left[\frac{x}{\sqrt{2}}\Biggm| \begin{array}{c}
   --    \\
(0, 1), \left(\frac{\upsilon + \alpha+2}{2}, \frac{1}{4}\right), \left(\frac{\upsilon - \alpha + 1}{2}, \frac{1}{4}\right)  \\
\end{array} \right].
  \end{equation}
\end{cor}

\begin{thm}\label{Th2}
  Let $\Re(\upsilon)>0, \Re(\gamma)>0, z>0, x>0, t>0, \delta>0$ and $0<\alpha \leq 1.$ Then the following $H$-function representation in two variables is valid:
  \begin{equation}\label{H2}
    \mathcal{T}_{3, \alpha}(\upsilon, t, z, \delta, x, \gamma) = \frac{x^{\frac{\upsilon}{\gamma}}}{4\gamma}
    H^{\;1,0:2,0;1,0}_{\;0,1:0,2;0,1}\left[\begin{array}{c}\frac{t x^{\frac{1}{\gamma}}}{2}\\x^{\frac{\delta}{\gamma}} z \end{array}\,\Biggm|
\begin{array}{c}
            -- \\
             \left(\frac{\gamma-\upsilon}{\gamma}, \frac{1}{\gamma}, \frac{\delta}{\gamma}\right), \left(\frac{-2\alpha-1}{4}, \frac{1}{2}\right), \left(\frac{2\alpha+1}{4}, \frac{1}{2}\right), (0, 1)
  \end{array} \right].
  \end{equation}
\end{thm}
\begin{proof}
  By making use of the two-variable Mellin transform (\cite{MSH})
  \begin{equation*}
    \mathcal{M}_{\mathcal{T}_{3, \alpha}} (s) = \int_{0}^{\infty} t^{s_{1}-1} \int_{0}^{\infty} z^{s_{2}-1} \left\{\int_{0}^{\infty} y^{\upsilon-1} f(t, z) \, dy\right\} dz \, dt,
  \end{equation*}
  to the standard thermonuclear function (3.16), we have
  \begin{equation*}
    \mathcal{M}_{\mathcal{T}_{3, \alpha}} (s) = \int_{0}^{\infty} t^{s_{1}-1} \int_{0}^{\infty} z^{s_{2}-1} \left\{\int_{0}^{\infty} y^{\upsilon-1} \mathbb{K}_{\alpha + \frac{1}{2}} (ty)\, \exp{(-zy^{\delta}-xy^{-\gamma})} \, dy\right\} dz \, dt.
  \end{equation*}
  By utilising \cite[p.346, Eq.(3.381(4))]{Gr-Ry} and swapping the integrals because of the uniform convergence, we can solve
\begin{align}\label{S3}
&\mathcal{M}_{\mathcal{T}_{3, \alpha}} (s) = \frac{2^{s_{1}-2}\, x^{\frac{\upsilon-s_{1}-\delta s_{2}}{\gamma}}}{\gamma} \Gamma(s_{2}) \Gamma\left(\frac{2\alpha + 1}{4} + \frac{s_{1}}{2}\right) \notag \\
&\qquad \qquad \qquad \times \Gamma\left(\frac{s_{1}}{2} - \frac{2\alpha + 1}{4}\right) \Gamma\left(1- \frac{\upsilon}{\gamma} + \frac{s_{1}}{\gamma} + \frac{\delta}{\gamma} s_{2}\right),
\end{align}
where $\Re(\gamma)>0, \Re(s_{2})>0, \Re(\upsilon - s_{1} - \delta s_{2})>0$ and $\Re\left(\frac{\upsilon - s_{1} - \delta s_{2}}{\gamma}\right)<1.$ Employing the inverse Mellin transform in two variables to \eqref{S3}, the Mellin-Barnes integral we obtain
\begin{align}\label{S4}
    &\mathcal{T}_{3, \alpha} = \frac{x^{\frac{\upsilon}{\gamma}}}{4 \gamma} \frac{1}{(2 \pi \iota)^{2}} \int_{\mathcal{L}_{1}} \int_{\mathcal{L}_{2}} \Gamma(s_{2}) \Gamma\left(\frac{2\alpha + 1}{4} + \frac{s_{1}}{2}\right) \Gamma\left(\frac{s_{1}}{2} - \frac{2\alpha + 1}{4}\right)  \notag \\
&\qquad \qquad \qquad \times \Gamma\left(1- \frac{\upsilon}{\gamma} + \frac{s_{1}}{\gamma} + \frac{\delta}{\gamma} s_{2}\right) (x^{\frac{1}{\gamma}} \frac{t}{2})^{-s_{1}} (x^{\frac{\delta}{\gamma}} z)^{-s_{2}} ds_{1} \, ds_{2}.
  \end{align}
  Now comparing \eqref{S4} to the definition of $H$-function \eqref{H-function} yields the required result.
\end{proof}

It is clearly evident that \eqref{Cr1}, which is identical to \eqref{t10}, appears in the MacDonald extended thermonuclear rate of reaction. Thus, in the illustration of the $H$-function, the MacDonald rate reaction probability integral, we get
\begin{align}\label{S5}
  & ^{MD}\hat{r}_{m n} = (1 - \frac{1}{2} \delta_{m n}) N_{m} N_{n} \left(\frac{2}{M}\right)^{\frac{1}{2}} \left(\frac{2}{\pi}\right)^{\frac{3}{2}} \Gamma\left(\frac{1 - \alpha}{2}\right) \notag\\ &\qquad  \times \sum_{\upsilon = 0}^{2} \left(\frac{1}{2 K T}\right)^{-\upsilon + \frac{1}{2}} \frac{S^{(\upsilon)} (0)}{\upsilon !} H^{3,\;0}_{0,\;3}\left[\frac{x}{\sqrt{2}}\Biggm| \begin{array}{c}
   --    \\
(0, 1), \left(\frac{\upsilon + \alpha+2}{2}, \frac{1}{4}\right), \left(\frac{\upsilon - \alpha + 1}{2}, \frac{1}{4}\right)  \\
\end{array} \right]
\end{align}
where $0<\alpha \leq 1, \Re(\upsilon)>0$ and $x = \frac{\tau}{\sqrt{KT}},$ that $\tau$ is given in \eqref{r9}.

\section{Conclusion}
The present study has developed a rigorous closed-form analytic framework for non-resonant thermonuclear reaction rates by replacing the classical Maxwell-Boltzmann distribution with a generalized velocity and kinetic energy distribution formulated in terms of the MacDonald function - the modified Bessel function of the second kind. This substitution is physically justified by conditions prevailing in stellar interiors where local departures from hydrostatic equilibrium invalidate the standard Maxwellian description of particle kinetics. Within this framework, the non-resonant thermonuclear reaction rate integrals have been systematically extended across four physically and mathematically distinct regimes - the standard, cut-off, depletion, and screened cases - each formulated in the center-of-mass system via the MacDonald-function-based relative velocity distribution, thereby furnishing a self-consistent and unified analytic treatment of the astrophysical thermonuclear functions. A central contribution of this work is the derivation of Mellin transform representations of these extended thermonuclear functions, which facilitated their closed-form expression in terms of the Fox $H$-function, the most general class of hypergeometric-type functions in the theory of generalized special functions, yielding representations that are simultaneously compact, analytically tractable, and asymptotically accessible. Explicit evaluations of the reaction probability integrals have further been carried out for several functional forms of the slowly varying astrophysical cross-section factor $S(E)$, with the extended formalism shown to recover its classical Maxwell-Boltzmann counterparts in the appropriate limiting cases, thereby confirming the internal consistency of the theoretical development.

The results obtained herein carry significant implications for the quantitative modeling of thermonuclear processes in dense and non-equilibrium stellar environments, encompassing the cores of main-sequence stars, red giants, and compact objects, where Coulomb screening, particle depletion, and energy cut-off effects are physically consequential. The analytic machinery established through the combined use of the Fox $H$-function and Mellin transform techniques constitutes a powerful and extensible mathematical apparatus, naturally admitting generalizations to resonant reaction mechanisms, multi-particle interaction channels, and higher-order quantum corrections to nuclear cross-sections. Moreover, the distributional framework introduced here provides a principled platform for incorporating non-extensive statistical distributions into the study of thermonuclear kinetics in cosmological settings, including primordial nucleosynthesis under conditions of non-equilibrium plasma evolution in the early universe. Future investigations will be directed toward establishing explicit numerical benchmarks of the extended reaction rates against standard stellar structure codes and toward quantifying the astrophysical consequences of MacDonald-distribution-based thermonuclear rates within the broader context of stellar evolution modeling.

%%%%%%%%%%%%%%%%%%%%%%%%%%%%%%%%%%%%%%%%%%%%%%%%%%%%%%%%%%%%%%%%%%%%%%%%%%%%%%%%%%%

\backmatter

\bmhead{Acknowledgments}
I sincerely thank Professor Dr. Rakesh K. Parmar for providing valuable insights and guidance throughout this research. Special thanks are also extended to the reviewer for their insightful feedback and helpful recommendations.

\section*{Declarations}

\bmhead{Author contribution}
Author contribution: The author wrote the paper.

\bmhead{Funding}
No funding was received to assist with the preparation of this manuscript.

\bmhead{Conflict of Interest}
The author declares no competing interests.

\bmhead{Data Availability}
No datasets were generated or analysed during the current study.

\bmhead{Materials Availability}
Not Applicable.

\bmhead{Ethics approval and consent to participate}
Not Applicable.

\bmhead{Consent for publication}
Not Applicable.

                           %%
%%===========================================================================================%%
\bibliography{ThermoNucle-Bibiliography}% common bib file
%% if required, the content of the .bbl file can be included here once bbl is generated
%%\input sn-article.bbl

\end{document}